\documentclass{llncs}

\usepackage[utf8]{inputenc}
\usepackage[T1]{fontenc}

\usepackage[font=footnotesize,labelfont=bf]{caption}
\usepackage{graphicx}
\usepackage{wrapfig}
\usepackage{subcaption}

\usepackage[pdftex,pdfpagelabels]{hyperref}
\hypersetup{%
  hidelinks
}
\newcommand{\aref}[1]{\hyperref[#1]{Appendix~\ref{#1}}}

\usepackage{thmtools,thm-restate}

\graphicspath{{./figs/}}

\let\originalparagraph\paragraph
\renewcommand{\paragraph}[2][.]{\originalparagraph{#2#1}}

\newcommand{\SK}{\mathcal{S}}
\newcommand{\SKd}{\mathcal{S}_d}
\newcommand{\WF}{\mathcal{W}}
\newcommand{\schoenT}{\mathcal{T}}

\usepackage{xcolor}
\definecolor{darkgreen}{rgb}{0,0.5,0}

\definecolor{darkpink}{rgb}{0.91, 0.33, 0.5}

\begin{document}

\mainmatter

\title{Representing Directed Trees as Straight Skeletons}
\titlerunning{Directed Trees as Straight Skeletons}

\author{Oswin Aichholzer\inst{1}
  \and
  Therese Biedl\inst{2}
  \and
  Thomas Hackl\inst{1}
  \and
  Martin Held\inst{3}
  \and
  Stefan Huber\inst{4}
  \and
  Peter Palfrader\inst{3}
  \and
  Birgit Vogtenhuber\inst{1}
  \thanks{%
    OA and BV supported by Austrian Science Fund (FWF) I 648-N18;
    TB by NSERC;
    TH by FWF P23629-N18;
    MH and PP by FWF P25816-N15.}%
}

\authorrunning{Oswin Aichholzer,
  Therese Biedl,
  Thomas Hackl,
  Martin Held,
  Stefan Huber,
  Peter Palfrader, and
  Birgit Vogtenhuber}

\institute{
    Technische Universität Graz,
    Institut für Softwaretechnologie,
    8010 Graz, Austria;
    \texttt{\{oaich,thackl,bvogt\}@ist.tugraz.at}
  \and
    David R.~Cheriton School of Computer Science,
    University of Waterloo, Waterloo, Ontario N2L 1A2, Canada;
    \texttt{biedl@uwaterloo.ca}
  \and
    Universität Salzburg,
    FB Computerwissenschaften,
    5020 Salzburg, Austria;
    \texttt{\{held,palfrader\}@cosy.sbg.ac.at}
  \and
    Institute of Science and Technology Austria,
    3400 Klosterneuburg, Austria;
    \texttt{stefan.huber@ist.ac.at}.
}

\maketitle

\begin{abstract}
The straight skeleton of a polygon is the geometric graph obtained by
tracing the vertices during a mitered offsetting process. It is known
that the straight skeleton of a simple polygon is a tree, and one can
naturally derive directions on the edges of the tree from the
propagation of the shrinking process.

In this paper, we ask the reverse question:  Given a tree with directed
edges, can it be the straight skeleton of a polygon?  And if so, can we
find a suitable simple polygon?  We answer these questions for all
directed trees where the order of edges around each node is fixed.
\end{abstract}


\section{Introduction}

Many geometric structures on sets of points, line segments, or polygons,
e.g.  Delaunay triangulations, Voronoi diagrams, straight skeletons, and
rectangle-of-influence graphs can be represented as graphs.  The
\textit{graph representation} problem (for each of these geometric
structures) asks which graphs can be represented in this way. That is,
given a graph $G$, can we find a suitable input set $S$ of points,
segments, or polygons such that the geometric structure induced by $S$
is equivalent to $G$?

Graph representation has been studied for numerous geometric structures
in the past.  To name just a few examples:  Every planar graph is the
intersection graph of line segments~\cite{CG09}, every wheel is a
rectangle-of-influence graph~\cite{LLMW98}, and all 4-connected planar
graphs are Delaunay triangulations~\cite{DS96}.  See
also~\cite{DBLL-GD94} for many results on proximity drawability of
graphs.

Of particular interest to our paper are two results.  First, Liotta and
Meijer~\cite{LM03} studied when a tree can be represented as the Voronoi
diagram of a set of points, and showed that this is always possible (and
the points are in convex position).  Secondly, Aichholzer et
al.~\cite{ACD+12} studied when a tree can be represented as the straight
skeleton of a polygon, and showed that this is always possible (and the
polygon is convex).

\subsection{Background} The \emph{straight skeleton} $\SK(P)$ of a
simple polygon $P$ is defined via a wavefront-propagation process:  Each
edge of $P$ emits a wavefront edge moving in a self-parallel manner at
unit speed towards the interior of the polygon.

Initially, at time $t=0$, this wavefront is a single polygon that is
identical to~$P$.
As the propagation process continues, however, the wavefront will change
due to self-interaction:
(i) In \emph{edge events}, an edge of the wavefront shrinks to zero
length and is removed from the wavefront.
(ii) In \emph{split events}, a vertex of the wavefront meets the
interior of a previously non-incident wavefront edge.  This split
partitions the edge and the polygon into two parts that now propagate
independently.
(iii) If the input is not in general position, more complex interactions
are possible.  For example, entire portions of the wavefront collapse at
once when parallel wavefronts that were emanated by parallel polygon
edges meet, or new reflex wavefront vertices are created when multiple
reflex vertices interact in a \emph{vertex-event}.
The propagation process ends once all components of the wavefront have
collapsed.
Therefore, the set of wavefront edges at any time $t$ form one or more
polygons, which we call the wavefront and denote by $\WF(t)$.

\begin{wrapfigure}{R}{0.47\textwidth}
  \centering
  \includegraphics[page=1]{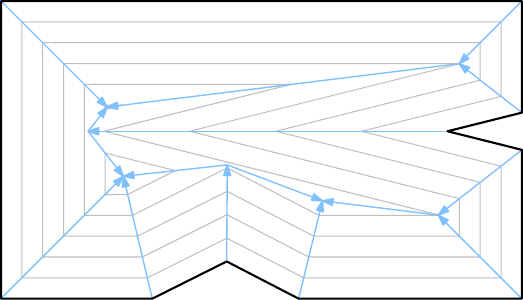}
  \caption{
    The straight skeleton $\SK(P)$ of an input polygon $P$ (bold)
    is the union of the traces of wavefront vertices.  Wavefront
    polygons at different times are shown in gray.
  }
  \label{fig:simple-sk}
\end{wrapfigure}

The straight skeleton $\SK(P)$, introduced by Aichholzer et
al.~\cite{Aic&95}, is then defined as the geometric graph whose edges
are the traces of the vertices of $\WF(t)$ over time; see
\autoref{fig:simple-sk}.  For simple polygons, $\SK(P)$ always is a
tree~\cite{Aic&95}, with the leaves corresponding to vertices of $P$ and
interior vertices having degree 3 or more.  Several algorithms are known
to construct the straight skeleton \cite{AiAu98,EpEr99,HuHe12a}.

We can distinguish between convex and reflex vertices of $P$ or
$\WF(t)$.  A vertex~$v$ is \emph{reflex} (\emph{convex}) if the interior
angle at $v$ is greater (less) than $\pi$.
We call an arc of $\SK(P)$ \emph{reflex} (\emph{convex}) if it was
traced out by a reflex (convex) vertex of the wavefront.  When
discussing the wavefront propagation process, we will often
interchangeably refer to wavefront vertices and straight skeleton arcs.

The roof model~\cite{Aic&95} represents a convenient means to study the
wavefront over the entire propagation period.  It embeds the wavefront
in three-space, where the $z$-axis represents time: $\schoenT(P) :=
\bigcup_{t \ge 0} ( \WF(t) \times \{t\})$.  The inner edges and vertices
of this polytope correspond to arcs and nodes of the straight skeleton
$\SK(P)$, and the $z$-coordinate of each element corresponds to the time
it was traced out by the wavefront-propagation process.
Reflex arcs correspond to valleys and convex arcs to ridges.

If we exclude polygons where parallel polygon edges cause entire
wavefront segments to collapse at one time, resulting in horizontal roof
edges, then arcs of $\SK(P)$ will have been traced out by the wavefront
during its propagation process, and we can assign a natural direction to
these arcs: make them point into their trace direction.  This assignment
gives rise to the \emph{directed straight skeleton}, $\SKd(P)$.

\smallskip

A \emph{directed tree} $T$ is a directed graph whose underlying
undirected graph is a tree, i.e., connected and acyclic.
A \emph{labeled tree} $T_\ell$ is a tree with assignments of labels to
its arcs.
For most of this paper, trees are \emph{ordered}, i.e., for every node
there is a fixed circular order in which the arcs appear around this
node.

It is customary to refer to the edges and vertices of the straight
skeleton as \emph{arcs} and \emph{nodes}, and to reserve \emph{edges}
and \emph{vertices} for elements of input or wavefront polygons.
We also use \emph{arcs} and \emph{nodes} to refer to elements of trees.

\subsection{Our results}

Our paper was inspired by the work in~\cite{ACD+12}, which studies
undirected trees.  However, the structure of the straight skeleton
imposes directions on the arcs, except in degenerate cases.  Hence, the
natural question to ask is:

\begin{problem}[Directed straight-skeleton realizability]
  Given a directed tree $T$,
  (i) is there a polygon $P$ such that $\SKd(P)$ shares the structure of
  $T$ (we denote this by $\SKd(P) \sim T$), and
  (ii) if yes, can we reconstruct such a polygon $P$ from $T$?
\end{problem}

Having directions assigned to the arcs makes the straight-skeleton
realizability problem significantly harder: For example, one easily sees
that in a convex polygon the straight skeleton is a rooted tree (with
exactly one sink), and so not all directed trees can be represented via
convex polygons. Hence, the results from~\cite{ACD+12} do not transfer
to directed trees.

The directed-straight-skeleton-realizability question can be asked for
multiple meanings of ``directed tree'':
It could be a {\em geometric tree} (nodes are given with coordinates),
an {\em ordered tree} (the clockwise order of arcs around each node is
specified) or an {\em unordered tree} (we have the nodes and arcs but
nothing else).
For a geometric tree, the problem is trivial, since the locations of the
leaves specify the vertices of the only polygon for which this could be
the straight skeleton.  (If  leaves are not specified as points  but
only as ``being on a ray'', then the geometric setting is non-trivial,
but can be solved in polynomial time~\cite{BieHH13}.)

In this paper, we consider the variant of the problem for ordered trees.
In the case of polygons in general position, we give three obviously
necessary conditions and show that these are also always sufficient.  It
turns out that the order of arcs around nodes is not important, so the
algorithm also works for unordered directed trees.
We then turn to polygons without restrictions on vertex-positions.  In
this case the directed straight skeletons can be significantly more
complicated, and in particular, have arbitrary degrees.  Testing whether
a directed tree can be represented as straight skeleton requires deeper
insight into the structure of straight skeletons, and we can exploit
these to develop such a testing algorithm and, in case of a positive
answer, find a suitable polygon.


\section{Trees from Polygons in General Position}

In a first step we restrict the problem to polygons in general position.
By general position we mean that no four edges have supporting lines
which are tangent to a common circle.  In particular this means that
during the wavefront propagation process only standard edge and split
events are observed, resulting in straight skeletons where all interior
nodes are of degree exactly three.

Investigating the structure of such directed straight skeletons enables
us to establish a number of necessary conditions for a directed tree to
be a directed straight skeleton of a polygon in general position.

\paragraph{Necessary conditions}

Let $P$ be a polygon in general position and let $T$ be the directed
tree such that $\SKd(P) \sim T$.

The leaves of $T$ correspond to the vertices of $P$.  In the roof model,
these vertices have zero $z$-coordinate, while all other nodes have
positive $z$-coordinates since they will have been created by an event
at some time $t > 0$.  Thus, any arc incident to a vertex $v$ of $P$
increases in elevation as it moves away from $v$.  As such, all leaves
of $T$ must have in-degree 0 and out-degree 1.

The interior nodes all have degree 3 and are classified by their
in- and out-degrees as follows:
\begin{description}
  \item[in-degree 3:] \emph{(peak nodes)}
    A collapse of a wavefront component (of triangular shape) at the end
    of its propagation process is witnessed by a local maximum in the
    roof.  These local maxima correspond to nodes with in-degree three.
  \item[in-degree 2:] \emph{(collapse nodes)}
    Edge events in the propagation process, i.e., collapsed wavefront
    edges, result in a node with two incoming arcs and one outgoing arc.
  \item[in-degree 1:] \emph{(split nodes)}
    Split events will cause a node that has only one incoming arc and
    two outgoing arcs.
  \item[in-degree 0:]
    Since the roof model will have no local minima except at the edges
    of $P$~\cite{Aic&95}, nodes with in-degree zero and out-degree three
    cannot exist.
\end{description}

Of these, the case of a split event requires some more attention since
it imposes additional restrictions on the incoming arc.
Recall that we can distinguish between reflex and convex vertices of
$P$, and note that any vertex is either reflex or convex by the general
position assumption.  For a split event to occur, a reflex vertex of the
wavefront must crash into a previously non-incident part of the
wavefront.

In the absence of vertex events, which create skeleton-nodes of degree
at least four and therefore do not happen when the polygon is in general
position, no reflex vertex can ever be created by an event.  Thus, any
reflex vertex that is part of an event must have been emanating from a
reflex vertex of the input polygon itself.
Accordingly, the incoming arc in a split event node must have a leaf at
its other end.

We summarize these conditions in the following lemma.

\begin{lemma}
  \label{lem:necessary}
  Let $P$ be a simple polygon in general position and let $T$ be the directed
  tree such that $\SKd(P) \sim T$.  Then in $T$
  \begin{description}
    \item[(G1)] the incident arc of each leaf is outgoing,
    \item[(G2)] every interior node has degree 3 and at most two
                outgoing arcs,
    \item[(G3)] if an interior node has out-degree two, then the
                incoming arc connects directly from a leaf. 
  \end{description}
\end{lemma}

\paragraph{Observations}  Let $T$ be a directed tree that satisfies
conditions (G1--G3).  Classify the interior nodes as split nodes,
collapse nodes and peaks as above.  The goal is to show that any such
tree can indeed be realized as a straight skeleton.
For this, we split the tree into multiple subtrees in a particular way
(also illustrated in the example in \autoref{fig:merge-polys}). We have
the following observations. Full proofs of the next four lemmas can be
found in
\aref{sec:proofs}.

\begin{restatable}{lemma}{LEMonepeak}
  \label{lem:one-peak}
  Let $T$ be a directed tree that satisfies conditions (G1--G3).
  If $T$ has no split nodes, then $T$ has exactly one peak.
\end{restatable}

\begin{restatable}{lemma}{LEMsplitTree}
  \label{lem:splitTree}
  Let $T$ be a tree that satisfies conditions (G1--G3).
  Create a forest $F$ as follows:  At any split node $s$ of $T$, remove
  $s$, remove the leaf incident to the incoming arc of $s$, and replace
  the two outgoing arcs of $s$  by two new leaves that are connected to
  the other ends of these arcs.
  Then each component of $F$ satisfies conditions (G1--G3) and has
  exactly one peak.
\end{restatable}

\paragraph{Sufficient conditions}

It remains to be shown that the necessary conditions (G1--G3) from
\autoref{lem:necessary} are also sufficient.  We show this by
constructing a simple polygon $P$ such that $\SKd(P) \sim T$, given a
directed tree $T$ that satisfies (G1--G3).  We start by showing this for
trees that have no split nodes.

\begin{wrapfigure}{R}{0.47\textwidth}
  \centering
  \includegraphics[page=1]{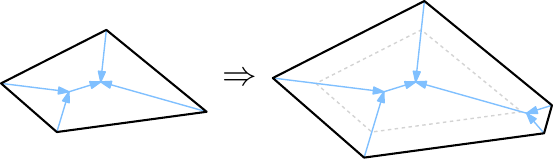}
  \caption{
    Creating the polygon for a tree with $k+1$ interior nodes
    from a polygon for a tree with one less.
  }
  \label{fig:creating-P}
\end{wrapfigure}

\begin{restatable}{lemma}{LEMsimpleExtend}
  \label{lem:simple-extend}
  For any directed tree $T$ that satisfies (G1--G3) and has no split
  nodes, there is a convex polygon $P$ such that $\SKd(P) \sim T$.
\end{restatable}

\begin{proof}
  We show this by constructive induction.  Any triangle is a polygon
  such that its straight skeleton shares the structure of the peak node
  of $T$.

  To construct a polygon $P$ for a tree $T$ with $k$ interior nodes, we
  first construct a polygon $P'$ for a tree $T'$ with $k-1$ interior
  nodes.  We obtain $T'$ by replacing a node of $T$ and its two adjacent
  leaves with a single leaf $\ell$.  (There always is such a node.)
  To obtain $P$, we compute an exterior offset of $P'$ and replace the
  vertex that corresponds to $\ell$ with a sufficiently small edge such
  that it collapses before the wavefront of $P$ reaches $P'$.

  This polygon $P$ will then satisfy $\SKd(P) \sim T$.
  See \autoref{fig:creating-P} for an illustration.
  \qed
\end{proof}

Furthermore, it is possible to add a constraint on one interior angle of
the polygon:

\begin{restatable}{lemma}{LEMsimpleExtendAngle}
  \label{lem:simple-extend-angle}
  Let $T$ be a directed tree without split nodes and let $\ell$ be a
  leaf of $T$.  Further, let $\alpha$ be an arbitrary angle with
  $0 < \alpha < \pi$.
  Then there exists a convex polygon $P$ such that $\SKd(P) \sim T$ and
  such that the interior angle at the vertex that corresponds to
  $\ell$ is $\alpha$.
\end{restatable}

Now we are ready to consider trees with split nodes.

\begin{lemma}
  Let $T$ be a directed tree that satisfies conditions (G1--G3).
  Then there exists a polygon $P$ such that $\SKd(P) \sim T$.
  \label{lem:constructing-P}
\end{lemma}

\begin{proof}
  As in \autoref{lem:splitTree}, we split $T$ at split nodes, also
  dropping the incoming reflex arc and its incident leaf.  We obtain a
  forest $F = \{T_1, T_2, \ldots, T_n\}$ where each $T_i$ is a tree
  without split nodes.

  This forest can in turn be considered an undirected tree, where each $T_i$
  gives rise to a node and nodes are connected if and only if the
  corresponding trees originally had a split node in common.  We pick an
  arbitrary root for $F$, say $T_1$, and construct a convex polygon
  $P_1$ such that $\SKd(P_1) \sim T_1$.

  This root $T_1$ is connected to one or more children in $F$ via split
  nodes.  Let $T_2$ be such a child, and let $\ell_1\in T_1$ and
  $\ell_2\in T_2$ be the two leaves obtained when splitting at the split
  node common to $T_1$ and $T_2$.  Let $v_1$ be the vertex in $P_1$
  corresponding to $\ell_1$, and assume it has angle $\alpha_1<\pi$.

  We construct a convex polygon $P_2$ such that $\SKd(P_2) \sim T_2$ and
  the vertex $v_2$ corresponding to $\ell_2$ has angle
  $\alpha<\pi-\alpha_1$.  This enables us to merge $P_1$ and $P_2$ in
  the following way:  We place $P_2$ in the plane such that $v_1$ of
  $P_1$ and $v_2$ of $P_2$ occupy the same locus.  We rotate $P_2$ such
  that the angle between a pair of edges of $P_1$ and $P_2$ is exactly
  $\pi$.  Which pair of edges is chosen depends on where in the cyclic
  order the incident reflex leaf at the split node in $T$ lies.  The
  layout of $P_1$ and $P_2$ then corresponds to the layout of the
  wavefront at the split-event time.  If we now compute a small outer
  offset and designate this to be $P$, then the directed straight
  skeleton of $P$ has the same structure as the subtree of $T$ that is
  made up by $T_1$, $T_2$, the split node, and its incident leaf; see
  \autoref{fig:merge-polys}.

  \begin{figure}[!ht]
    \captionsetup{aboveskip=0mm,belowskip=0ex}
    \centering
    \begin{subfigure}[b]{0.49\columnwidth}
      \centering
      \includegraphics[page=1]{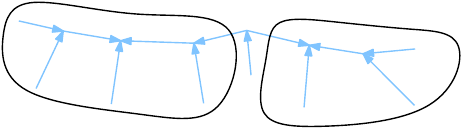}
      \caption{}
    \end{subfigure}
    \begin{subfigure}[b]{0.49\columnwidth}
      \centering
      \includegraphics[page=2]{merge}
      \caption{}
    \end{subfigure}
    \caption{
      (a) Given a directed tree we split it into a forest $F$ of
      subtrees without split nodes.  (b) Recursing on the structure of
      $F$, we can create convex polygons for each element (dotted) and
      then merge them into ever larger polygons.
    }
    \label{fig:merge-polys}
  \end{figure}

  We then repeat this process for another child of $T_1$ or $T_2$.  Note
  that it may be necessary to scale the polygon that we add to a
  sufficiently small size so that it does not conflict with other parts
  of the polygon already constructed.  This is always possible since for
  each vertex of a polygon there exists a disk that intersects the
  polygon only in the wedge defined by the vertex.

  Once all elements of the forest have been processed, we have
  constructed a polygon $P$ whose straight skeleton has the same
  structure as $T$, i.e., $\SKd(P) \sim T$.
  \qed
\end{proof}

Notice that conditions (G1--G3) do not depend on the order of arcs
around nodes; we can construct a polygon for any such order.  So in
particular if $T$ is an unordered tree that satisfies (G1--G3), then we
can pick an arbitrary order and the lemma holds.  Hence, we get the
following theorem:

\begin{theorem}
  An (ordered or unordered) directed tree $T$ is the directed straight
  skeleton of a simple polygon $P$ in general position if and only if
  $T$ satisfies conditions (G1--G3).
\end{theorem}


\section{Realizing trees with labeled arcs}

Recall that an arc of the straight skeleton is called \emph{reflex}
(\emph{convex}) if it was traced out by a reflex (convex) vertex of the
wavefront.  One can easily see that the construction in
\autoref{lem:constructing-P} creates a polygon where all arcs of the
straight skeleton are convex, with the exception of arcs from leaves to
split nodes.

For later constructions (for trees with higher degrees), it will be
important that we test not only whether an ordered tree can be realized,
but additionally we want to impose onto each arc whether it is reflex or
convex in the straight skeleton.  We study this question here first for
trees with maximum degree 3.

So assume we have a directed tree $T$ that satisfies (G1--G3).
Additionally we now label each arc of $T$ with either ``reflex'' or
``convex'', and we ask whether there exists a polygon $P$ that realizes
this labeled directed tree in the sense that $\SKd(P)\sim T$ and the
type of skeleton-arc in $\SKd(P)$ matches the label of the arc in~$T$.
We denote this by $\SKd(P)\sim_\ell T$.

We observe that a peak node is created when a wavefront of three edges,
a triangle, collapses.  Therefore, all incident arcs at a peak node are
convex.

For collapse nodes, we know that the outgoing arc is convex.  (Recall
that reflex arcs in a straight skeleton are only created in vertex
events, which cannot exist when all interior nodes have degree three.)
At least one of the incoming arcs needs to be convex, as two reflex
wavefront vertices meeting in an event will result in a node of degree
at least four.

We have already established that the incoming arc at a split node needs
to be reflex.  Furthermore, it is easy to see that the two outgoing arcs
of a split node need to be convex.

We summarize the necessary conditions for a labeled directed tree to
correspond to a straight skeleton in the following lemma:

\begin{lemma}
  Let $P$ be a simple polygon  in general position and let $T_\ell$ be
  the labeled directed tree such that $\SKd(P) \sim_\ell T_\ell$.  Then
  \begin{description}
    \item[(L1)] for peak nodes, all incoming arcs are convex;
    \item[(L2)] for collapse nodes, at least one incoming arc and the outgoing arc are convex;
    \item[(L3)] for split nodes, the incoming arc is reflex and both outgoing arcs are convex.
  \end{description}
  \label{lem:labeled-necessary}
\end{lemma}

We will now show that (L1--L3) are also sufficient:
\begin{lemma}
  Any labeled directed tree $T_\ell$ that meets conditions (G1--G3) from
  \autoref{lem:necessary} and (L1--L3) from
  \autoref{lem:labeled-necessary} is realizable by a simple polygon $P$.
\end{lemma}

\begin{wrapfigure}{R}{0.47\textwidth}
  \centering
  \includegraphics[page=1]{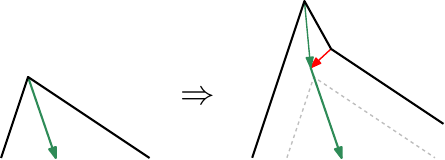}
  \caption{
    Extending a convex vertex of $P$ such that a leaf in its tree is
    replaced by a collapse node where one incoming arc is reflex and one
    is convex.
  }
  \label{fig:z-shape}
\end{wrapfigure}

\begin{proof}
  Since reflex arcs in $T_\ell$ only originate at leaves and terminate
  at collapse or split nodes (interior nodes never have outgoing reflex
  arcs, and peak nodes have no incoming reflex arcs), we can create a
  tree $T'_\ell$ by replacing each collapse node that has an incoming
  reflex arc with a leaf and dropping the two incident incoming arcs and
  their leaves.  This resulting tree $T'_\ell$ will have no reflex arcs
  except for those leading from a leaf to split nodes by (L1), (L3), and
  (G3). Thus, we can create a polygon $P'$ such that $\SKd(P') \sim_\ell
  T'_\ell$ by the process described in the proof of
  \autoref{lem:constructing-P}, respecting all labels.

  We now obtain $P$ by offsetting $P'$ slightly to the outside.  Then,
  we modify $P$ at each vertex $v$ that corresponds to a leaf in
  $T'_\ell$ that was the result of replacing a collapse node of
  $T_\ell$.  Note that each such vertex $v$ is convex since the outgoing
  arc of a collapse node is convex by (L2).  We insert a small edge in
  place of $v$, replacing it with $v_1$, the edge, and $v_2$.  We choose
  the angles at $v_1$ and $v_2$ such that one of them is reflex and the
  other is convex, in order to match the labeling of $T_\ell$.
  \autoref{fig:z-shape} illustrates this operation.

  By making the new edge sufficiently small, we can ensure that these
  events happen before the wavefront becomes identical to $P'$, and thus
  before all remaining events of the wavefront propagation.
  \qed
\end{proof}


\section{Arbitrary Node Degrees}

Once we allow for straight skeletons where interior nodes can have
degrees larger than three, a number of previous constraints no longer
hold.  Most importantly, during the wavefront propagation vertex events
can happen, resulting in new reflex vertices in the wavefront after the
event.  Consequently, for instance, split nodes no longer need to be
adjacent to leaves.  Larger node degrees also result in more complex
variants of split, collapse, and peak nodes.
Note that we continue to restrict polygons from having parallel edges as
those might cause skeleton arcs which have no direction (when they get
created as a result of two wavefront edges crashing into each other) or
straight skeleton arcs that are neither reflex nor convex (when two
parallel wavefronts moving in the same direction become incident at an
event).

\smallskip

In order to understand what combinations of reflex and convex incoming
and outgoing arcs may exist at a node in a directed straight skeleton,
we study the different shapes that a wavefront may have at an event at
locus $p$ and time $t$.
At a time $t-\delta$ immediately prior to the event, the wavefront will
consist of a combination of reflex and convex vertices, tracing out
reflex and convex arcs, all of which will meet at $p$ at time $t$.  We
choose $\delta$ sufficiently small such that no event will happen in the
interval $[t-\delta, t)$.

Consider the wavefront around a locus $p$ at an event, and consider the
\emph{wedges} that have been already swept by the wavefront.  With wedge
we mean the area near $p$ swept over by a continuous portion of the
wavefront polygon until time $t$.

If a single wedge covers the entire area around $p$, we call it a
\emph{full} wedge.  The interior angles of other wedges might be less
than $\pi$, greater than $\pi$, or exactly $\pi$ as illustrated in
\autoref{fig:wedges}.  We call the first type of wedge \emph{reflex} and
the second type of wedge \emph{convex}, after their corresponding
wavefront vertices in the simple case.  The third type of wedge we
simply call $\pi$-wedge.

\begin{figure}[!ht]
  \captionsetup{belowskip=-0.8\baselineskip}
    \centering
    \includegraphics[page=1]{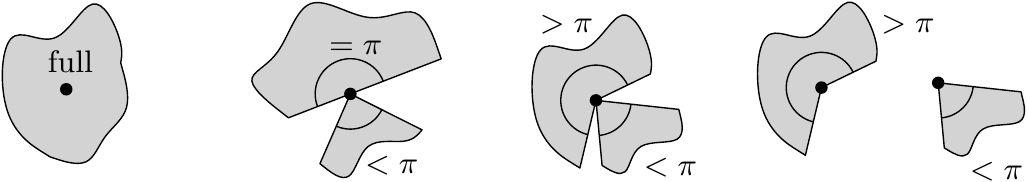}
    \caption{
      Wavefront wedges at event times are either full wedges or can be
      classified by their interior angle.  The wedges (already-swept
      areas) are gray, the remaining white sectors are covered by the
      wavefront polygon(s).
    }
    \label{fig:wedges}
\end{figure}

A single wedge may have been traced out by just one wavefront edge if a
wedge at an event has an interior angle of exactly $\pi$, but in all
other cases it is the area covered by two or more wavefront edges and
their incident wavefront vertices, which have traced out one or more
incoming arcs at $p$.  Note that all but the two outermost edges of this
part of the wavefront collapse at time $t$.

\smallskip

We will establish arc-\emph{patterns} to describe combinations of arcs
at a node.  Such a pattern is a string consisting of the types of arcs:
$r$ for an incoming reflex arc, $c$ for an incoming convex arc, and
$\hat{r}$ and $\hat{c}$ for their outgoing counterparts.  We will use
operators known from language theory, such as parentheses to group
blocks, the asterisk $(^*)$ to indicate the preceding character or group
may occur zero or more times, the plus sign $(^+)$ to indicate the
preceding block may occur one or more times, and the question mark
$(^?)$ to indicate it may exist zero times or once.  When defining
patterns we give them variable names in capitals.

We now investigate which combination of arcs may trace out which types
of wedges.  For this purpose we first characterize single wedges and
provide their describing arc-patterns.  Then we examine all possible
single wedge combinations and provide allowed arc-patterns for interior
nodes.

In this section, we give only an overview of arc-patterns
of single wedges and of arc-patterns for possible combinations of these
wedges at interior nodes.  For a detailed analysis please see
\aref{sec:patterns}.

As mentioned, single wedges can be reflex wedges, convex wedges,
$\pi$-wedges, and so called \emph{full wedges}, the latter being traced
out by a wavefront that collapses completely around a locus $p$.  The
arc-pattern for a reflex wedge is $R := r ~ (c r)^*$, i.e., one reflex
vertex, followed by zero or more (convex, reflex) pairs of vertices and
thus arcs.  If a reflex wedge is created by a single reflex wavefront
vertex (and its incident edges) then we call it \emph{trivial}.
Otherwise, we call it \emph{non-trivial}, with $R_+ = r~(c r)^+$
specifying the arc-pattern of a non-trivial reflex wedge.

The arc-pattern for a convex wedge is $C := r^? ~ c ~ (r^? c)^* ~ r^?$.
Like for reflex wedges we distinguish trivial (traced out by a single
convex vertex) and non-trivial convex wedges, the latter (any pattern
matching $C$ and having length at least two) being denoted by $C_+$.

The case of $\pi$-wedges is related to the case of reflex wedges.  A
$\pi$-wedge is either traced out by exactly one wavefront edge (no
incoming arcs), or it has the same pattern as for a non-trivial reflex
wedge.  Hence, we have as arc-pattern $P := \emptyset ~|~ R_+$.  Note
that since we have explicitly excluded parallel polygon edges for this
problem setting, only trivial $\pi$-wedges can exist.  Therefore, we
will set $P := \emptyset$ here.

Last, the arc-pattern of a full wedge is
$F := c ~ (r^? c)^* ~ c ~ (r^? c)^* ~ c ~ (r^? c)^*$.

\begin{figure}[!ht]
  \centering
  \includegraphics[page=1]{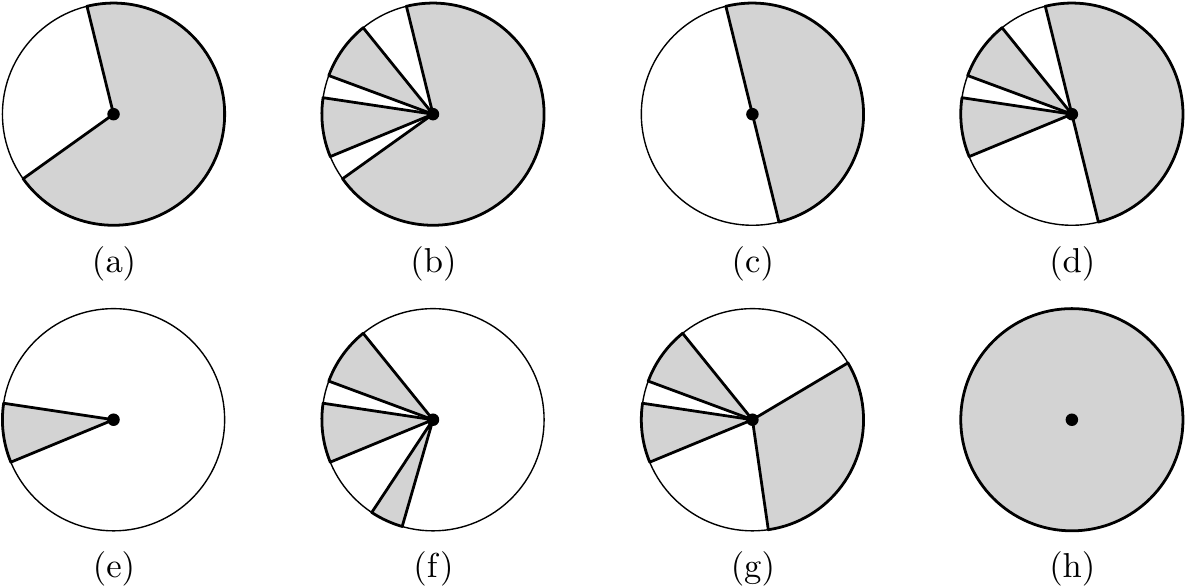}
  \caption{
    All wedge combinations possible at a node.
  }
  \label{fig:node-wedges}
\end{figure}

Analyzing possible combinations of single wedges at interior nodes we
get seven different allowed arc-patterns (see
\autoref{fig:node-wedges}~(a-h)):
$N_a := C_+ ~ \hat{c}$\ ,\quad
$N_b := C ~ \hat{c} ~ (R ~ \hat{c})^+$\ ,\quad
$N_d := P ~ \hat{c} ~ (R ~ \hat{c})^+$\ ,\quad
$N_e := R_+ ~ \hat{r}$\ ,\quad
$N_f := (R ~ \hat{c})^+ R ~ \hat{r}$\ ,\quad
$N_g := (R ~ \hat{c})^+ R ~ \hat{c}$\ ,\quad and
$N_h : = F$.
(Note that the pattern from \autoref{fig:node-wedges}~(c) will not
result in an event without parallel polygon edges, which we have
excluded.)
A simple split node is matched via $N_d$, a simple collapse node is
handled by $N_a$ (with $C_+ $ being either $cc$, $rc$, or $cr$), and a
simple peak node is matched by $N_h$.

Since these are all possible wavefront/wedge combinations, any interior
node of the straight skeleton will have to match
$N := N_a | N_b | N_d | N_e | N_f | N_g | N_h$.
We can state the following lemma:

\begin{lemma}
  \label{lem:full-necessary}
  Let $P$ be a simple polygon and let $T_\ell$ be the labeled directed
  tree such that $\SKd(P) \sim_\ell T_\ell$.
  Then the cyclic order of arcs of any interior node of $T_\ell$ needs
  to match the pattern specification $N$ defined above.
\end{lemma}

After the necessity of the discussed conditions we now prove their
sufficiency.

\begin{lemma}
  \label{lem:full-sufficient}
  Let $T_\ell$ be any labeled directed tree for which
  (i) the cyclic order of arcs of each interior node matches the pattern
      specification $N$ defined above and
  (ii) each leaf has out-degree one.
  Then $T_\ell$ is realizable by some simple polygon~$P$, that is,
  $\SKd(P) \sim_\ell T_\ell$.
\end{lemma}

\begin{proof}
  We will construct $P$ in a manner similar to the one described in
  \autoref{lem:constructing-P}.  We start by identifying maximally
  connected components $C_1, C_2, \ldots, C_n$ of $T_\ell$ containing
  nodes with out-degree either zero or one.  These components take the
  place of the subtrees of our forest from \autoref{lem:constructing-P},
  and they are connected in $T_\ell$ via split-nodes, i.e., nodes with
  out-degree $\ge 2$.

  We pick an arbitrary component $C_1$ and create a polygon
  $P_1$ such that $\SKd(P_1) \sim_\ell C_1$ as follows.
  If there are no outgoing arcs from $C_1$, we start at its unique peak
  node.  Otherwise, there is exactly one outgoing arc, and we begin at
  the node it is incident to.  Constructing a polygon for this
  node is straightforward by applying the concepts learned from
  considering convex, reflex, and full wavefront wedges.
  We proceed by extending this initial polygon step by step like in
  \autoref{lem:simple-extend}, treating each reflex or convex vertex of
  the polygon as a wavefront wedge to be constructed, until we have a
  polygon for the entire component.

  Next, we pick one of the split nodes connected to $C_1$.  That split
  node, we call it $n$, will have one of the forms from
  \autoref{fig:node-wedges} that have at least two white sectors.  The
  polygon we just created will cover one of these white sectors.
  (Depending on the type of arc connecting $C_1$ to $n$, the polygon
  will either have a reflex or convex vertex for this arc.)
  We continue by constructing polygons for all remaining white sectors
  in the same fashion we used for constructing $P_1$.  Note that this
  process allows us to force at least one angle, and therefore we can
  construct polygons that fit into the white sectors for $n$.

  Now we have the wavefront polygon as it should be when the event at
  $n$ happens.  We compute a small exterior offset, joining all polygons
  into one larger polygon.  The new reflex or convex vertices at this
  point are then further subdivided as required by the incoming arcs for
  $n$.

  We repeat this process until we have covered all split nodes and thus
  all components and have thereby created a polygon whose structure
  matches $T_\ell$.
  \qed
  Please see
  \aref{sec:sample-full-construction}
  for an example.
\end{proof}

\smallskip

Combining \autoref{lem:full-necessary} and
\autoref{lem:full-sufficient}, we obtain the following theorem:

\begin{theorem}
  An ordered labeled directed tree $T_\ell$ is the directed straight
  skeleton of a simple polygon $P$ without parallel edges if and only if
  (i) the cyclic order of arcs of each interior node of $T_\ell$ matches
  the pattern specification $N$ defined above and
  (ii) each leaf has out-degree one.
\end{theorem}

\paragraph{Conclusion}
In this work we developed a complete characterization of necessary and
also sufficient conditions such that a given directed and labeled
ordered tree can be represented as the straight skeleton of a simple
polygon. This extends previous work on representing trees via related
geometric structures~\cite{ACD+12,LM03}.

We leave the algorithmic question -- how efficient suitability of a
given input tree can be tested and, in case of an affirmative answer, a
corresponding simple polygon can be computed -- for future research. We
conjecture that both is possible in time linear in the size of the given
tree.

\bibliography{bibs/names-short,bibs/weasel,bibs/extra}
\bibliographystyle{splncs03}

\newpage


\appendix
\section{Proofs}
\label{sec:proofs}

\subsection{Proof of \autoref{lem:one-peak}
  (See page~\pageref{lem:one-peak})}
\LEMonepeak*
\begin{proof}
  Say $T$ has $k$ nodes and $k-1$ directed arcs.  Since (G1--G3) hold,
  each node has out-degree exactly zero or one.  There are $k-1$ arcs,
  so there are $k-1$ nodes with out-degree one.  This leaves exactly one
  node with out-degree zero, the single peak node.
  \qed
\end{proof}

\subsection{Proof of \autoref{lem:splitTree}
  (See page~\pageref{lem:splitTree})}
\LEMsplitTree*
\begin{proof}
  Removing one $s$ gives three connected components.  One of these is a
  leaf and removed.  In the other two, we added a leaf with one outgoing
  arc, and for all other nodes in-degrees and out-degrees are unchanged.
  So (G1--G3) hold for all created components.  Since out-degrees are
  unchanged, every component of the final forest $F$ has no split nodes,
  and hence a unique peak by the \autoref{lem:one-peak}.
  \qed
\end{proof}

\subsection{Proof of \autoref{lem:simple-extend}
  (See page~\pageref{lem:simple-extend})}
\LEMsimpleExtend*

This is a more verbose version of the proof:

\begin{proof}
  We show this constructively by induction on the number of interior
  nodes of $T$.  Recall that $T$ has a unique peak node.

  In the base case, $T$ consists just of a peak node and three incident
  leaves.  Choose an arbitrary triangle; the straight skeleton then has
  such a structure.

  Now let us assume that for every tree $T$ with $k$ interior nodes we
  can find a convex polygon $P$ such that $\SKd(P) \sim T$.

  Given $T'$ with $k+1$ interior nodes we construct $P'$ with
  $\SKd(P') \sim T'$ as follows.
  Find an interior node $n$ of $T'$ that is incident to two leaves and
  obtain $T''$ from $T'$ by dropping these two leaves, turning $n$ in
  $T''$ into a leaf.  Since $T''$ has only $k$ interior nodes, we can
  find a convex polygon $P''$ such that $\SKd(P'') \sim T''$.

  To construct $P'$, we compute a small exterior offset of $P''$ and
  replace the vertex corresponding to $n$ with a small edge.  It should
  be sufficiently small that it will collapse before the wavefront
  propagating from $P'$ reaches $P''$.  This ensures that the event
  happens before any other in the wavefront propagation.

  By this construction, the polygon $P'$ will satisfy $\SKd(P') \sim T'$.
  See \autoref{fig:creating-P} for an illustration.
  \qed
\end{proof}

\subsection{Proof of \autoref{lem:simple-extend-angle}
  (See page~\pageref{lem:simple-extend-angle})}
\LEMsimpleExtendAngle*

\begin{proof}
  Let $p$ be the unique peak node in $T$, and let
  $p=v_0,v_1,\dots,v_{m-1},v_m=\ell$ be the nodes on the unique path
  from $p$ to $\ell$ in $T$.  Pick a series of $m$ values $\alpha_i$
  such that
  $0 < \alpha_1 < \alpha_2 < \ldots < \alpha_{m-1} < \alpha_m = \alpha$.

  When we create the triangle whose straight skeleton corresponds to the
  tree immediately around $p$, we choose it so that the angle at the
  polygon vertex corresponding to $v_1$ is $\alpha_1$.  Throughout the
  expansion steps, we maintain that the angle at the vertex
  corresponding to $v_i$ (if it is a leaf of the current polygon) is
  $\alpha_i$.

  When expanding the polygon at some leaf, we can choose the angle at
  one of the incident new vertices, as long as it is larger than the
  angle of the vertex replaced and less than $\pi$.  See
  \autoref{fig:creating-P-with-angle} for an illustration.

  Thus, if we expand at the vertex corresponding to $v_i$ (for some
  $i\geq 1$), it has angle $\alpha_i$, and we can pick the angle at the
  vertex corresponding to $v_{i+1}$ to be $\alpha_{i+1}$.  At the end of
  the expansion the angle at the vertex that corresponds to $v_m=\ell$
  is now $\alpha_m = \alpha$.
  \qed
\end{proof}

\begin{figure}[!ht]
  \centering
  \includegraphics[page=1]{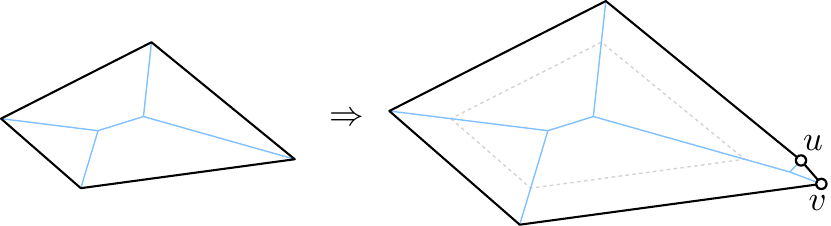}
  \caption{
    When extending a polygon, we can choose the angle at one of the two
    new vertices.  The new angle can be an arbitrary value between the
    angle of the replaced vertex and $\pi$.
  }
  \label{fig:creating-P-with-angle}
\end{figure}


\newpage
\section{Details on allowed arc-patterns at interior nodes}
\label{sec:patterns}

We investigate which combination of arcs may trace out which types of
wedges.  For this purpose it is useful to look at the wavefront at a
time $t-\delta$ immediately before the event, when all edges are still
of proper length and not collapsed. We first consider single wedges and
will afterwards discuss all possible combinations of them.
Since all edges move at unit speed and reach $p$ exactly at time $t$, we
know that just before the event their supporting lines must all be
tangent to a common circle.

\paragraph{Reflex Wedges}

First, let us look at reflex wedges (\autoref{fig:wedge-reflex}).  A
reflex wedge can, of course, be created by a single reflex wavefront
vertex and its incident edges.  However, if it is created by more edges,
then:
\begin{itemize}
  \item
    The vertices incident to the two outermost edges ($e_\ell$ and $e_r$
    in \autoref{fig:wedge-reflex}) need to be reflex, because adding a
    convex vertex would introduce a wavefront edge that has already
    swept $p$ by time $t$.
  \item
    Since no two consecutive reflex vertices of the wavefront can become
    incident at an event (they always move away from one another), there
    needs to be a convex vertex between any two reflex vertices.  This
    convex vertex will also reach $p$ at time $t$.
  \item
    Between two reflex vertices of a reflex wedge, no two consecutive
    convex vertices can exist.  If they did, the wavefront edge between
    them at time $t-\delta$ would not be tangent to the circle and thus,
    would collapse prior to $t$, causing an event.  However, we chose
    $\delta$ such that there is no event between $t-\delta$ and~$t$.
  \item
    Any reflex vertex of the reflex wedge at time $t-\delta$ can in turn
    be replaced by a pair of reflex vertices with their convex vertex
    connector; cf. \autoref{fig:wedge-reflex}~(3).  This implies that an
    arbitrary number of reflex vertices can appear in one reflex wedge.
\end{itemize}

\begin{figure}[!ht]
  \centering
  \includegraphics[page=1]{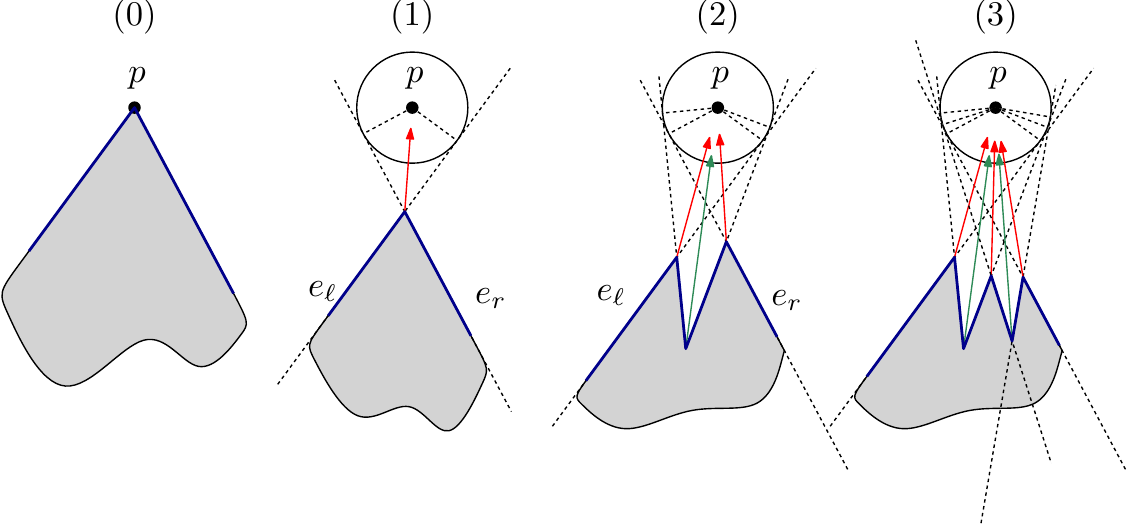}
  \caption{
    Wavefronts (bold, blue) of a reflex wedge at (0) and immediately
    before (1--3) an event.
    Area already swept over by the wavefront is shown in gray.  All
    edges participating in the event have supporting lines that are
    tangent to the circle around the event location.
  }
  \label{fig:wedge-reflex}
\end{figure}

Summarizing, the pattern describing the arcs of a reflex wedge is $R :=
r ~ (c r)^*$, i.e., one reflex vertex, followed by zero or more (convex,
reflex) pairs of vertices and thus arcs.

We call a reflex wedge \emph{trivial} if it was traced out by a single
vertex and its two incident wavefront edges, and \emph{non-trivial}
otherwise.  Let $R_+ = r ~ (c r)^+$ be the pattern specification for a
non-trivial reflex wedge.

\paragraph{Convex Wedges}

A convex wedge in the simplest case is the result of a single convex
vertex and its incident wavefront edges.  More complex situations are
possible:  Many consecutive convex vertices may reach $p$ at the same
time.  Further, between any two convex vertices there may be one reflex
vertex (see e.g.  \autoref{fig:z-shape} on how to replace a convex
vertex with two vertices, one of them reflex). Even the outermost
vertices in any collapsing chain need not be convex.

\begin{figure}[!ht]
  \captionsetup{belowskip=-0.8\baselineskip}
  \centering
  \includegraphics[page=1]{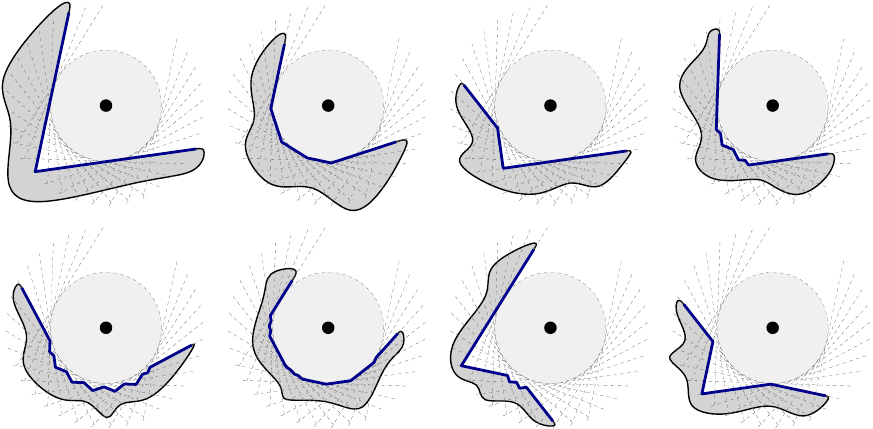}
  \caption{
    Wavefronts (bold, blue) of a convex wedge immediately before an
    event.
    Area already swept over by the wavefront is shown in gray.  All
    edges participating in the event have supporting lines that are
    tangent to the circle around the event location.
  }
  \label{fig:wedge-convex}
\end{figure}

Therefore, we can establish the pattern $C := r^? ~ c ~ (r^? c)^* ~ r^?$
for a convex wedge.

Similar to before, we call a convex wedge \emph{trivial} if it was
traced out by a single vertex and its two incident wavefront edges and
\emph{non-trivial} otherwise.  We use $C_+$ to denote any pattern
matching $C$ and having length at least two, i.e., a non-trivial convex
wedge.

\paragraph{$\pi$-wedges}

A wedge with interior angle $\pi$ is related to the case of the reflex
wedge.  Either it is traced out by exactly one wavefront edge and there
are no incoming arcs from this wedge to the event node, or it has the
same pattern as for a reflex wedge.  This is easy to see when again
considering a time $t-\delta$ and the circle such that all participating
wavefront edges lie in supporting lines tangential to the circle.
Either it is just one edge, and it is tangential, or there are several
and they need to start and end with reflex vertices to be tangent to the
circle.  Since in that case no wavefront edge actually touches the
circle, we again cannot have more than one consecutive convex vertex.

\begin{figure}[!ht]
  \centering
  \includegraphics[page=1]{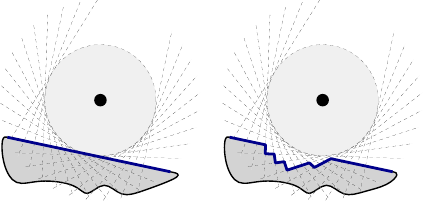}
  \caption{
    Wavefronts (bold, blue) of a $\pi$-wedge immediately before an
    event.
    Area already swept over by the wavefront is shown in gray.  All
    edges participating in the event have supporting lines that are
    tangent to the circle around the event location.
  }
  \label{fig:wedge-pi}
\end{figure}

We call a $\pi$-wedge \emph{trivial} if it was traced out by a single
wavefront edge and, hence, has no incoming arc, and \emph{non-trivial}
otherwise.

Thus, the pattern for a $\pi$-wedge is $P := \emptyset ~|~ R_+$, that
is, either it is empty or it is a non-trivial reflex-wedge pattern,
i.e., with at least two reflex vertices.

Note that the outgoing arc of a non-trivial $\pi$-wedge is traced out by
a vertex of the wavefront with interior angle exactly $\pi$.  In such
cases our problem statement that requires all arcs to be labeled either
reflex or convex becomes invalid.  We have therefore excluded polygons
with parallel edges, and so for the purposes of this paper, only trivial
$\pi$-wedges can exist.  Therefore, we will set $P := \emptyset$ here.

\paragraph{Full wedge}

A full wedge is traced out by a wavefront that collapses completely
around a locus $p$.  In the simplest case this is a triangle.  As in the
case of the convex wedge we can replace any convex wavefront vertex with
more convex wavefront vertices, and we can also put reflex vertices
between any two convex vertices.  The only constraint is that there need
to be at least three convex vertices in total in order to close the
wavefront polygon.

\begin{figure}[!ht]
  \centering
  \includegraphics[page=1]{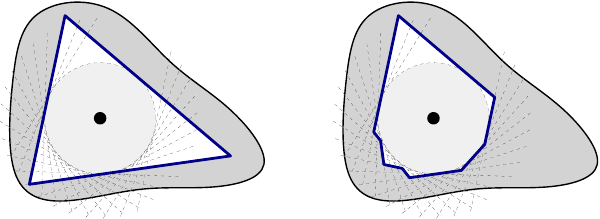}
  \caption{
    Wavefronts (bold, blue) of a full wedge immediately before an event.
    Area already swept over by the wavefront is shown in gray.  All
    edges participating in the event have supporting lines that are
    tangent to the circle around the event location.
  }
  \label{fig:wedge-full}
\end{figure}

Therefore, the pattern for a full wedge is
$F := c ~ (r^? c)^* ~ c ~ (r^? c)^* ~ c ~ (r^? c)^*$.

\bigskip

\paragraph{Wedge combinations at events}

Let us now consider which wedge combinations may be present at an event,
and which outgoing wavefront vertices can be created in each case, cf.
\autoref{fig:node-wedges}.

So we have now defined patterns $R,R_+,C,C_+,P$, and $F$ for wedges that
are reflex, convex, $\pi$-wedges, or full.  Around each node, there may
be many wedges, but at most one of them can be convex, $\pi$-wedge, or
full.

\begin{description}
  \item[one convex wedge:]~\\
    If there is exactly one convex wedge, then it needs to be a
    non-trivial wedge.  (Otherwise there would be no event.)  The event
    will create a single convex outgoing wavefront vertex.  The full
    pattern specification for this event therefore is
    $N_a := C_+ ~ \hat{c}$.
    See \autoref{fig:node-wedges}~(a).
  \item[one convex wedge and one or more reflex wedges:]~\\
    The $1+n$ wavefront wedges at the event leave a total of $1+n$
    uncovered (white) sectors that the wavefront will soon sweep.  For
    each of those, a single outgoing arc is created, and it must be
    convex since each unswept sector has an angle of less than $\pi$.
    $N_b := C ~ \hat{c} ~ (R ~ \hat{c})^+$.
    See \autoref{fig:node-wedges}~(b).
  \item[one $\pi$-wedge:]~\\
    For this to be an event, the $\pi$-wedge needs to be non-trivial.
    The single vertex outgoing from this event will be connecting two
    wavefront edges sharing a common supporting line and will thus be
    neither reflex nor convex.  We have excluded such cases in our
    problem statement.
    See \autoref{fig:node-wedges}~(c).
  \item[one $\pi$-wedge and one or more reflex wedges:]~\\
    This is an extended variant of the classic split event.  The node's
    pattern is
    $N_d := P ~ \hat{c} ~ (R ~ \hat{c})^+$.
    See \autoref{fig:node-wedges}~(d).
  \item[one reflex wedge:]~\\
    The wedge has to be non-trivial, else there would be no event.  So
    this is a vertex event, creating a new reflex vertex.
    $N_e := R_+ ~ \hat{r}$.
    See \autoref{fig:node-wedges}~(e).
  \item[two or more reflex wedges all in the same half-plane:]~\\
    Similar to the previous case, this must be a vertex event creating a
    new reflex vertex in the empty section whose interior angle is
    maximal (and thus larger than $\pi$).  In all other sections a
    convex vertex is created.
    $N_f := (R ~ \hat{c})^+ ~ R ~ \hat{r}$.
    See \autoref{fig:node-wedges}~(f).
  \item[two or more reflex wedges not all in the same half-plane:]~\\
    Since there is no empty half-plane around the event, a new convex
    vertex will be created for each empty sector.
    $N_g := (R ~ \hat{c})^+ R ~ \hat{c}$.
    See \autoref{fig:node-wedges}~(g).
  \item[full wedge:]~\\
    If the complete disk around an event has already been swept by the
    wavefront, then a wavefront polygon is collapsing about the event
    location and we have a full wedge.
    $N_h : = F$.
    See \autoref{fig:node-wedges}~(h).
\end{description}

Since these are all possible wavefront/wedge combinations, any interior
node of the straight skeleton will have to match
$N := N_a | N_b | N_d | N_e | N_f | N_g | N_h$.
Here, ``matching'' means matching cyclically, i.e., if a tree has a node
where we find in cyclic order, a reflex outgoing, a reflex incoming, a
convex incoming, and a reflex incoming, i.e., $\hat{r} ~ r ~ c ~ r$,
then this would match $N_e = R_+ ~ \hat{r}$ since we end up with two
identical patterns by rotating one string cyclically.


\section{Sample construction process (\autoref{lem:full-sufficient})}
\label{sec:sample-full-construction}

Suppose we have a split node $n$ matching the pattern
$r ~ c ~ r ~ \hat{c} ~ \hat{c}$, where each of those arcs has some more
tree components on the side not at $n$; see \autoref{fig:example1}~(a).

We see that it is a split event, matching the \emph{one $\pi$-wedge and
one or more reflex wedges} case with
$N_d := P ~ \hat{c} ~ (R ~ \hat{c})^+$.
Here, $P$ is the trivial $\pi$-wedge of $\emptyset$ and the only reflex
wedge is $r ~ c ~ r$; see \autoref{fig:example1}~(b).

When we get to this split node, we already have created the polygon for
one of the components $C_1$ or $C_2$.  Assume we did $C_1$ already and
produced $P_1$ (\autoref{fig:example1}~(c)).  We next create a polygon
$P_2$ for $C_2$, making sure that the angle at $n$ is sufficiently small
to fit at $n$.  In this example it needs to be less than $\pi - \alpha$
if $\alpha$ is the angle at $n$ for $P_1$.

\begin{figure}[!ht]
  \centering
  \begin{subfigure}[b]{0.32\columnwidth}
      \centering
      \includegraphics[page=1]{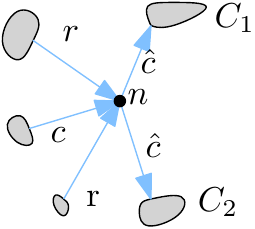}
      \caption{}
  \end{subfigure}
  \begin{subfigure}[b]{0.32\columnwidth}
      \centering
      \includegraphics[page=2]{example}
      \caption{}
  \end{subfigure}
  \begin{subfigure}[b]{0.32\columnwidth}
      \centering
      \includegraphics[page=3]{example}
      \caption{}
  \end{subfigure}
  \caption{
    A split node of a tree (a) and its wedge representation (b).
    A possible polygon $P_1$ for $C_1$ (c).
  }
  \label{fig:example1}
\end{figure}

Now that we have all the polygons present at the event node $n$, we
place them on the plane such the vertices of $P_1$ and $P_2$ that
correspond to the arc from $n$ become coincident in the same locus $p$.
Since this is a split node with a trivial $\pi$-wedge we also need to
rotate $P_2$ such that it forms an angle of exactly $\pi$ with $P_1$ on
one side.  Note that it may be necessary to scale $P_2$ such that it
does not intersect anything already constructed.

Next, we compute a small offset of the merged polygons.  The vertex $v$
illustrated in \autoref{fig:example2}~(a) currently stands for the
entire reflex wedge.  However, that wedge originates from three arcs, $r
~ c ~ r$, so we split that piece of the wavefront accordingly, see
\autoref{fig:example2}~(b).

\begin{figure}[!ht]
  \centering
  \begin{subfigure}[b]{0.49\columnwidth}
    \centering
    \includegraphics[page=4]{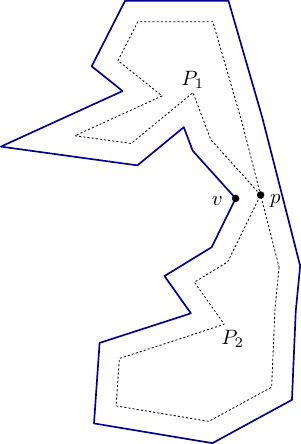}
    \caption{}
  \end{subfigure}
  \begin{subfigure}[b]{0.49\columnwidth}
    \centering
    \includegraphics[page=5]{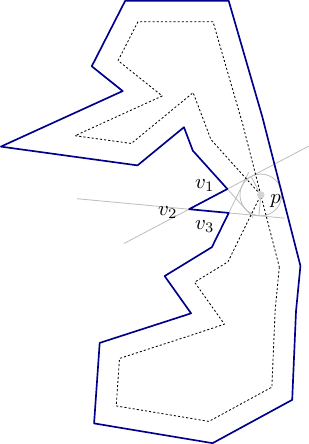}
    \caption{}
  \end{subfigure}
  \caption{
    First we combine the two polygons for the individual components and
    compute an outer offset.  Then we replace the vertex $v$ for the
    reflex wedge by a set of vertices tracing out all the arcs of the
    wedge.
  }
  \label{fig:example2}
\end{figure}

\end{document}